\documentclass[11pt]{article}
\usepackage{graphicx,graphics,color}
\usepackage{amsmath,amssymb, bm}
\usepackage{amsthm}
\usepackage{mathrsfs}
\usepackage{mathabx}
\usepackage{algorithm}
\usepackage{algpseudocode}
\usepackage{setspace}
\usepackage[utf8]{inputenc}
\usepackage{natbib}%
\usepackage[dvipsnames]{xcolor}
\usepackage{multirow}
\usepackage{caption}
\usepackage{subcaption}

\usepackage{listings}
\usepackage{color}

\definecolor{dkgreen}{rgb}{0,0.6,0}
\definecolor{gray}{rgb}{0.5,0.5,0.5}
\definecolor{mauve}{rgb}{0.58,0,0.82}

\usepackage[colorlinks=true,linkcolor=blue,citecolor=blue]{hyperref}%

\usepackage[margin=1.0 in, top = 1.0in]{geometry}

\newcommand{\cL}{\mathcal{L}}

\newcommand{\cT}{\mathcal{T}}
\newcommand{\cU}{\mathcal{U}}

\newcommand{\cX}{\mathcal{X}}
\newcommand{\cY}{\mathcal{Y}}

\newcommand{\cW}{\mathcal{W}}

\newcommand{\R}{{\rm I}\kern-0.18em{\rm R}}
\newcommand{\h}{{\rm I}\kern-0.18em{\rm H}}
\newcommand{\K}{{\rm I}\kern-0.18em{\rm K}}
\newcommand{\p}{{\rm I}\kern-0.18em{\rm P}}
\newcommand{\E}{{\rm I}\kern-0.18em{\rm E}}
\newcommand{\Z}{{\rm Z}\kern-0.18em{\rm Z}}
\newcommand{\1}{{\rm 1}\kern-0.24em{\rm I}}
\newcommand{\N}{{\rm I}\kern-0.18em{\rm N}}

\newcommand{\argmin}{\mathop{\mathrm{argmin}}}

\def\@begintheorem#1#2{\trivlist \item[\hskip \labelsep{\bf #1\ #2.}]\sl}
\def\@opargbegintheorem#1#2#3{\trivlist
      \item[\hskip \labelsep{\bf #1\ #2\ (#3).}]\sl}

\newtheorem{theorem}{Theorem}[section]
\newtheorem{corollary}{Corollary}[theorem]
\newtheorem{lemma}[theorem]{Lemma}
\newtheorem{assumption}{Assumption}

\title{Automatic Double Machine Learning for Continuous Treatment Effects\thanks{I thank Alberto Abadie, Victor Chernozhukov, Lindsey Currier, Vitor Hadad, Claire Lazar Reich, Anna Mikusheva, Stephen Morris, Whitney Newey, Victor Orestes, Eitan Sapiro-Gheiler, Rahul Singh, Sophie Sun, Rafael Veil, Jaume Vives for helpful discussions and the participants of MIT econometrics lunch seminar for helpful comments}}
\author{Sylvia Klosin\footnote{email: klosins@mit.edu}}
\date{\today}

\begin{document}
\maketitle

\begin{abstract}

In this paper, we introduce and prove asymptotic normality for a new nonparametric estimator of continuous treatment effects.  Specifically, we estimate the average dose-response function \textemdash the expected value of an outcome of interest at a particular level of the treatment level. We utilize tools from both the double debiased machine learning (DML) and the automatic double machine learning (ADML) literatures to construct our estimator. Our estimator utilizes a novel debiasing method that leads to nice theoretical stability and balancing properties. In simulations our estimator performs well compared to current methods.


\end{abstract}

\paragraph{}

\newpage

\section{Introduction}
 In this paper we propose a new nonparametric estimator of continuous treatment effects and prove asymptotic normality. Continuous treatment effects are of significant importance in applied economics. For example, labor economists study how the number of hours in a job training program impacts worker earnings \citep{flores2007estimating}, and political economists seek to understand how distance to polling locations impacts propensity to vote \citep{cantoni2020precinct}\footnote{Empirical examples from additional recent applied economics papers are given in Appendix \ref{section:more_emperical_examples}}.

There are many interesting statistical objects associated with continuous treatment effects. The object we specifically focus on estimating in this paper is the expected value of the outcome variable at a given level of the treatment variable. This object is also known as an average dose-response function or the average structural function \citep{blundell2001endogeneity}. Under conditional unconfoundedness assumptions, our parameter has a causal interpretation, and if conditional unconfoundedness does not hold, the parameter is still of descriptive interest.

Current applied work often estimates continuous treatment effects by imposing linear functional form assumptions. This may lead to treatment effect estimates that are biased and hard to interpret. To avoid parametric functional form assumptions, we suggest a nonparametric estimator. Our nonparametric estimator incorporates machine learning (ML) based on Lasso.

Though ML methods like Lasso produce accurate predictions, they induce bias in the parameter estimates. This ``regularization'' bias is due to the bias-variance trade-off inherent in minimizing the mean square error by the ML algorithms. Given the high variance of the ML methods, to keep mean square error low the algorithms must accept bias increases. This impacts not only the mean of the estimator, but estimation of its asymptotic distribution as well. Therefore simply estimating standard objects of interest by plugging in ML methods (called ``plug-in'' estimators) can lead to obstacles in statistical inference. For instance, using an off-the-shelf Lasso estimate of the regression function and plugging into the first step of the estimation procedure will not produce valid confidence intervals \citep{chernozhukov2016locally}. There are also inherent model selection problems that arise when using ML.

This paper aims to avoid these inference problems by adapting techniques for incorporating ML methods into estimators while preserving desirable statistical properties of these estimators. Specifically we build on work in \cite{chernozhukov2016locally} and \cite{chernozhukov2018double} that propose solutions to ``debias'' machine learning methods (DML). One important step in debiasing estimators is through sample splitting: averaging over observations of the data different than those used to estimate the models. Another important debiasing step is adding a debiasing term to the equation that defines the parameter of interest. 

In the initial DML literature the structure of debiasing terms is found analytically and then estimated by ML methods. However, in some settings the structure is not known, and even when it is, this estimation procedure might be undesirable. In these cases we can instead estimate these debiasing terms through a new procedure called automatic double debiased machine learning (Auto-DML) that estimates the debiasing term directly \citep{chernozhukov2018automatic}. For the estimands of interest in this paper --- points on the dose-response curve --- directly estimating the debiasing term is undesirable, for reasons described below, and we will instead use Auto-DML. However, tools from the current Auto-DML framework cannot be immediately applied to our continuous treatment effect setting given the localization around treatment values that estimation requires.

The contribution of this paper is to adapt the tools of Auto-DML toward the problem of estimating continuous treatment effects.  The papers that have studied the DML framework for continuous treatment effects include  \cite{su2019non} and \cite{colangelo2020double}. This paper is most closely related to the work of  \cite{colangelo2020double}, which presents and proves statistical properties of the DML version of the estimator we study here. A key difference between our estimators is how the debiasing term is constructed. Different versions of debiased continuous treatment effect estimators are given by \cite{kennedy_non-parametric_2017} and \cite{kallus2018policy}.

In our continuous treatment setting, the structure of the debiasing term is known, but there are several reasons we would still want to estimate continuous treatment effects with Auto-DML rather than DML. Specifically, the debiasing term includes the multiplicative inverse of the generalized propensity score (MIGPS), which is the multiplicative inverse of the probability density of the treatment variable given the covariates ($\frac{1}{f(t|X)}$)\footnote{Readers may be more familiar with the binary treatment case where the equivalent is Inverse Probability-of-Treatment Weighting (IPW) or Horvitz Thompson weights $\frac{1}{e(X)}$ where $e(x) = \p(T = 1| X)$. }. Current DML estimators involve plugging in an estimated MIGPS and then inverting it, which can be numerically unstable.  Our Auto-DML approach estimates the MIGPS directly instead. We incorporate this estimated term additively into the bias correction term rather than inversely.

In addition, directly estimating the MIGPS has desirable balancing properties. Dealing with balancing and trimming is one of the key components of applied work, and Auto-DML is tailor-made for balancing in comparison to DML. Current DML approaches for continuous treatment effects must conduct implicit trimming by using bounded kernels to achieve numerical stability; however, with trimming, we are no longer conducting inference on the original population, but instead on a sub-population of some form. There is limited theoretical justification for ad hoc trimming and censoring of this type \citep{crump2009dealing}. In comparison, Auto-DML does not require trimming through a kernel, and has been shown to work well without much additional trimming; e.g., in the binary IV treatment case \citep{singh2019biased}. Empirically we find through Monte Carlo simulations that when we avoid implicit trimming our estimator decreases root mean square error (RMSE) --- up to 50\% in certain specifications --- compared to current methods.

The issue of balancing and overlap is especially important in the continuous case. The number of people close to a fixed value $t$ with particular values of the covariates may be very small. Therefore, it has been noted that researchers must be careful when using inverse propensity weighting for continuous treatments because the effects may be exquisitely sensitive to the specification of conditional density \citep{hernan2010causal}. A motivation for this paper is to see if the advantages of Auto-DML over DML in the binary case transfer over to the continuous case.

The paper will continue in the following way. Section \ref{section:setup} sets up the framework of the paper and introduces the parameter of interest. The assumptions needed to identify this parameter of interest are given in Section \ref{section:identification}. In Section \ref{section:learningproblem} we describe the learning problem for estimating our parameter, which includes our novel MIGPS estimation procedure. Section \ref{section:asumptoticnormality} provides the theoretical results. Simulation results are described in Section \ref{section:numerical_examples}. Section \ref{section:conclusion} concludes.

\section{Setup} \label{section:setup}

\subsection{Notation}

We assume we have independent and identically distributed data $(W_1, \cdots, W_n)$ where the $W_i = (X_i, T_i, Y_i)$ are copies of a random variable $W$ with support $\cW = \cX \times \cT \times \cY$, with a cumulative distribution function (cdf) $F_{YTX}(Y,T,X)$. We use capital letters to denote random variables and lowercase letters to denote their possible values. For each unit in a large population $X_i \in \R^{p}$ denotes a vector of covariates, with $p$ potentially large, and $T_i \in \R$ as the continuous treatment.  We use the potential outcome framework \citep{rubin1974estimating} and $Y_i(t)$ denotes the potential outcome that would have been observed for individual $i$ under treatment level $t$. 

To simplify presentation let $\gamma_0(t, x)$ denote the true conditional expectation function $\E[Y| T = t, X = x]$. We also use $\alpha_0(t, X)$ to denote the true multiplicative inverse of the generalized propensity score (MIGPS) $\frac{1}{f_{T|X}(t|X)}$. As will be explained later, these two functions\footnote{they are functions and we stick to the language of the literature and call them nuisance parameters. } will be our nuisance parameters. They are called nuisance parameters because though we need them in order to estimate of parameter of interest, they themselves are not inherently of interest.

We use the $|\cdot|_q$ as the $\ell_q$ norm of a vector, and denote by $\|\cdot\|$ the $\cL_2$ norm of a random variable i.e. $\|X_i\| = \sqrt{\E[X_i]^2} $ . Latter when have random matrices

We denote the kernel function by $K_h(x) = \frac{1}{h}K(\frac{X - x}{h})$\footnote{We primarily use the gaussian kernel, and in the Section \ref{section:numerical_examples} we discuses the epanechnikov kernel as well}. Here $h$ is the bin-width of the kernel. Let the roughness of the kernel $K$ we are using be denoted by $R(K) := \int_{-\infty}^{\infty} K(u) du$ as in section 2.2 of \cite{hansen2009lecture}.

When we use sample splitting on the data, the full sample will be split in $L$ different folds. Let  $W_{\ell}$ for $\ell \in 1:L$ denote the data that is the $\ell$-th fold. Let $W_{\ell}^c$ denote the data that is the complement of the $\ell$-th fold \footnote{so if $L = 5$, then $W_1^c = W_2 \cup W_3 \cup W_4 \cup W_5 $}.

\subsection{Parameter of Interest} \label{section:parameterofinterest}

For a fixed value $t$ the goal is to estimate

\begin{equation}
\label{eq:parameter_of_interest}
    \beta_{t} = \E[Y(t)] = \E[\gamma_0(t, X)].
\end{equation}

This object is also known as an average dose-response function. It is also known as the average structural function \citep{blundell2001endogeneity}.

 This parameter of interest $\beta_{t}$ is implicitly defined by the moment function $g$

\begin{equation}
\label{eq:moment_condition}
g(W, \beta, \gamma_{0}) = \gamma_0(t, X) - \beta  
\end{equation}

\begin{equation}
\E[g(W, \beta, \gamma_0)] = 0 \text{ iff } \beta = \beta_{t}
\end{equation}

 Here $\gamma_0(t,X)$ is the true conditional expectation regression function, and it is a nuisance parameter that must be estimated in order to estimate the parameter of interest.

\paragraph{Example}

Now we define a $\beta_{t}$ for an example data generating process (DGP) in order to help the reader better understand our object of interest. Let's say we have a covariate $X_1 \sim N(1,1)$ and noise $\epsilon \sim N(0, 1)$ and treatment $T \sim N(2,1)$ and our outcome $Y = e^{T} + 3 X_1 + \epsilon$. Let's say we take $t = 0$, and so we are interested in $\beta_{0}$. In this case $\gamma_0(0, X) = e^{0} + 3X_1$ and so $\E[\gamma_0(0, X)] = e^{0} + 3 \E[X_1] = 1 + 3 = 4$. Hence in this toy example $\beta_0 = 4$

\section{Identification}  \label{section:identification}

Our parameter of interest $\beta_{t}$ is a function of potential outcomes which are not observed in the data directly. Therefore, we need to make assumptions which enable us to use the observational data to do inference about $\E[Y(t)]$. The following characterization of $\beta_{t}$ will be useful,

\begin{equation}
  \beta_{t} = \int_X \E[Y|T = t, X] dF_x(X) = \E[\gamma(t, X)] =  \lim_{h \rightarrow 0} \int_{T} \int_{Y} \int_X \frac{K_{h}(T - t) Y}{f_{T|X}(t|X)} dF_{Y,T, X}(Y, T, X).
\end{equation}

The assumptions outlined in this section are the ones required for identification of causal effects. They must hold for every $t \in \cT$ that a researchers wants to do inference on. As mentioned in \cite{kennedy_non-parametric_2017}, when a researcher has randomized experimental data, this assumptions hold. With observation data these assumptions are harder to justify, and generally impossible to impossible to test, however the estimate is still useful. The estimator still gives us an adjusted measure of association that is interesting in its own right. 

\begin{assumption} (Identification)
\label{assump:identification} 
\begin{enumerate}
    \item (Conditional unconfoundedness)  $Y(t) \perp T | X \quad $
    \item  (Overlap) For any $t \in \cT$ and $X \in \cX$, $f_{T|X}(t|X)$ is bounded away from zero \label{assump:overlap} 

    \item (Consistency) $T = t$ implies $Y = Y(t)$ \label{assump:consistency}
\end{enumerate}

\end{assumption}

Conditional unconfoundedness assumption is also known as the ``ignorability'' assumption. It is also related to the exogeneity assumptions made in applied economics work.  The assumptions intuitively means that controlling for covariates $X$, the treatment level $t$ is effectively random. See \cite{imbens_role_2000} for weaker form of this.

Overlap is also known as the ``positivity'' assumption. We are assuming that the propensity score is uniformly bounded away from 0 for all values in the support of the pre-treatment variables. As discussed in \cite{imbens_role_2000}, when one has a continuous treatment this may be harder to satisfy than the more commonly studied binary treatment case. This condition also becomes harder to satisfy when the dimensionality of the covariates is large.

Consistency is a causal assumption that is not always explicitly stated, but is basically always assumed in some form. Consistency means that the observed outcome for individuals with treatment level $t$ equals her outcome if she had received treatment $t$ \citep{hernan2010causal}.

\section{Learning Problem} \label{section:learningproblem}

\subsection{Debiased Moment}

Now we are ready to explain our proposed estimation procedure. Above in equation \eqref{eq:moment_condition} we gave the moment function that defined our parameter of interest $g(W, \beta, \gamma_{0}) = \gamma_0(t, X) - \beta  $. 

If a researcher wanted to estimate the parameter with ML, a first pass at the problem could be a ``plug-in'' approach. In such an approach, we would use ML to fit the model for nuisance parameter $\gamma(t,X)$, and then predictions of the model would be used to create our parameter $\beta_t$  according to the moment function. 

However this plug-in approach leads to large bias in our estimate of $\beta_t$ as explained in the introduction. The DML approach gives us a way to create a new debiased moment function which enables us to avoid the bias of the plug-in approach. We denote this new debiased moment function by $\psi$. 

\begin{equation}
\label{eq:debiased_moment_function}
    \begin{aligned}
         \psi(W_i, \beta_0, \gamma_0, \alpha_0) &= g(W_i, \beta_0, \gamma_0) + \phi(W_i, \beta_0, \gamma_0, \alpha_0) \\
         &= \gamma_0(t, X_i) - \beta_0 + K_h(T_i - t) \alpha_0(t, X)(Y_i - \gamma_0(t, X_i))
    \end{aligned}
\end{equation}

We call $\phi(w, \beta, \gamma, \alpha)$ our ``debiasing term''. As with our original moment function  \eqref{eq:moment_condition}, it is also a function of $\gamma(t, X)$, but now we introduce a new second nuisance parameter $\alpha(t, X)$ corresponding to the multiplicative inverse of the propensity score (MIGPS). In the continuous treatment effect setting of our paper we know that the correct debiasing term is $\alpha_0(t, X) = \frac{1}{f(t|X)}$. Our paper will be following the Auto-DML literature and estimating $\alpha$ directly \citep{chernozhukov2018automatic}. This is in contrast to the current literature, in which $\hat{f}(t|X)$ is estimated as a function of $X$ and then $\hat{\alpha}(t,X) = \frac{1}{\hat{f}(t|X)}$.

Note that our debiasing term $\phi(w, \beta, \gamma, \alpha)$ is a function of a kernel $K_h(T_i - t)$, where we are localizing around the specific treatment level we are interested in $t$. As explained in \cite{colangelo2020double}, as the bin-width of the kernel $h \rightarrow 0$ we have Neyman orthogonality as defined in \citep{chernozhukov2016locally} \citep{neyman1959optimal}.

\subsubsection{Estimation} \label{subsection:estimator}

We use the empirical analog of the debiased moment function \eqref{eq:debiased_moment_function} as our estimator.

\begin{equation}
\label{eq:estimator}
    \hat{\beta}_t = \frac{1}{n} \sum_{\ell=1}^{L} \sum_{i \in \ell} \hat{\gamma}_{\ell}(t,X_i) +  K_h(T_i - t) \hat{\alpha_{\ell}}(t, X_i)(Y_i - \hat{\gamma}_{\ell}(t, X_i))
\end{equation}

To construct this estimate of $\beta_t$ in practice there are two different estimation stages: stage 1 for estimating nuisance parameters $\hat{\gamma}$ and $\hat{\alpha}$, and stage 2 for estimating the parameter of interest $\hat{\beta}_t$. Now we provide a brief outline of the process because going into the details of each step

\begin{enumerate}

\item[Stage 1]
\begin{itemize}
    \item[i] Start with data splitting. First pick the number of splits $L$, where $L \in \{2, \cdots, n\}$\footnote{Common default numbers of splits include $L = 5$ and $L = 10$}. Then partition the observations indices into the $L$ different groups. We use $\ell$ to denote these groups $\ell = 1, \cdots, L$. Denote observations in group $\ell$ by $W_{\ell}$
    \item[ii] For each fold $\ell$ estimate the nuisance parameters $\hat{\alpha}_{\ell}$ and $\hat{\gamma}_{\ell}$ 
\end{itemize}

    \item[Stage 2] 
    
    \begin{itemize}
        \item[iii] Using the nuisance parameters predicted on the left out folds construct the new debiased moment function $\psi$ to create our estimate of $\beta_t$ by summing across all observations in \eqref{eq:estimator}
    \item[iv] Calculate the variance using the new moment function

    \begin{equation}
    \label{eq:variance_estimator}
\begin{aligned}
\hat{V}_t &= \frac{h}{n} \sum_{\ell = 1}^L \sum_{i \in W_{\ell}} \hat{\psi}_{\ell i}^2  \\
&=\frac{h}{n} \sum_{\ell = 1}^L \sum_{i \in I_{\ell}} \bigg(\hat{\gamma}_{\ell}(t, X_i) - \hat{\beta}_t + K_h(T_i - t) \hat{\alpha}_{\ell}(t, X)(Y_i - \hat{\gamma}_{\ell}(t, X_i))  \bigg)^2
\end{aligned}
    \end{equation}
    
   \end{itemize}
\end{enumerate}

Now we go into details of the construction.

\paragraph{Stage 1}

Estimating  $\hat{\alpha}(t, X)$ and $\hat{\gamma}(t, X)$.

We approximate the value of the MIGPS at $t$ by a suitable linear combination of basis functions of the observed covariates, with coefficients localized at $t$. Let $b(x)$ be a $p \times 1$ dictionary of functions.\footnote{For example, when we simulate the estimator in section \ref{subsection:simulation study} we set $b(X)$ to be a fifth order polynomial set of the covariate variables.}

Our goal is to find a vector of coefficients $\hat{\rho}_t$ \footnote{Note that $\hat{\rho}_t$ has a $t$ subscript because the coefficient is for a specific level of the continuous treatment $t$} for our dictionary such that

$$\widehat{\frac{1}{f_{T|X}(t|X)} } :=  \hat{\alpha}(t, X) = b(X)\hat{\rho}_t.$$ 

As we'll discuss below, we can find the appropriate $\hat{\rho}_t$ by solving the following Lasso problem, 
\begin{equation}
    \hat{\rho}_{t} = \argmin_{\rho} \{ - 2 \hat{M}' \rho + \rho' \hat{Q}\rho + 2 r_L|\rho|_1 \} \quad  |\rho_t|_1 = \sum_{j=1}^p|\rho_j|,
\end{equation}
where $\hat{M}$ and $\hat{Q}$ are defined as
\begin{equation}
    \hat{M} = \frac{1}{n} \sum_{i = 1}^{n} b(X_i)\quad \text{and} \quad \hat{Q} = \frac{1}{n} \sum_{i = 1}^{n} K_h(T_i - t) b(X_i) b(X_i)'.
\end{equation}

Here $r_L$ is the Lasso regularization parameter.  Often, when running Lasso, $r_L$  is selected via cross-validation. However, given that we don't know what true treatment effect, we are not able to use cross validation to pick $r_L$ in a traditional way in our case. Therefore, we follow the iterative procedure in \cite{chernozhukov2018automatic}  to determine its value in practice.

Before continuing, we give some intuition for the structure of estimation problem. 

The goal is to  find $\hat{\rho}_t$ such that 

\begin{equation}
      b(X)\hat{\rho}_t = \widehat{\frac{1}{f_{T|X}(t|X)}} 
\end{equation}

If we actually observed $\frac{1}{f_{T|X}(t|X)}$ we could find the $\hat{\rho}_t$ that enables us to approximate it best by solving a weighted least squares problem of the following form, 
\begin{equation}
\label{eq:weighted_ols}
\begin{aligned}
       \rho_t &= \argmin_{\rho} \E[K_h(T - t)(\frac{1}{f(t|X)} - \rho'b(X))^2] \\ 
             &= \argmin_{\rho} - 2\rho \textcolor{blue}{ \E[K_h(T - t)\frac{b(X)}{f(t|X)}]} + \rho' \textcolor{red}{\E [K_h(T - t) b(X)b(X)']} \rho .
\end{aligned}
\end{equation}

approximate with

\begin{equation}
\label{eq:final_lasso_problem}
    \hat{\rho}_t = \argmin_{\rho} - 2\rho\textcolor{blue}{ \frac{1}{n} \sum_{i = 1}^{n} b(X_i)} + \rho' \textcolor{red}{\frac{1}{n} \sum_{i = 1}^{n} K_h(T_i - t) b(X_i) b(X_i)'} \rho 
\end{equation}

We formalize this intuition in the next Lemma.

\begin{lemma}
\label{lemma:approximate_weighted_ols}
The minimization problem in equation \eqref{eq:weighted_ols} can be approximated by the final minimization problem we solve on in equation \eqref{eq:final_lasso_problem}
\end{lemma}

Note from the first order conditions of the Lasso we get this following balancing equation

\begin{equation}
    |\frac{1}{n} \sum b(X_i) - \frac{1}{n} \sum K_h(T_i - t)b(X_i) b(X_i)' \rho_{L} | \leq r_L
\end{equation}

From this expression we see that $b(X_i)' \rho_{L}$ serve to approximately balancing the overall sample average with the sample average of the group with treatment value $T$ close to our $t$ of interest. This type of balancing condition for the binary treatment case is given in \cite{chernozhukov2018automatic}, \cite{athey2016approximate}, \cite{zubizarreta2015stable}.

Finally, to estimate $\hat{\gamma}(t,X)$ (the CEF), we similarly project $\gamma_0(t,X)$ onto a $p$-dimensional dictionary. We focus on using Lasso in this paper. Given that we make predictions $\hat{\gamma}(t, X)$ at a fixed point $t$, we require uniform rates, which are available for Lasso.

\paragraph{Stage 2}

From stage 1 we have functions $\hat{\alpha}_{\ell}(t,X)$ and $\hat{\gamma}_{\ell}(t,X)$ for each of the folds $\ell$. Recall that  $\hat{\alpha}_{\ell}(t,X)$ and $\hat{\gamma}(t,X)$ were fit using $W_{\ell}^C$, and now in stage two we use these functions to find fitted values for observations in $W_{\ell}$. Sum up all observations in all folds, and sum up all folds and divide by $n$ as in equation \eqref{eq:estimator}

Similarly, sum up over all observations at in equation \eqref{eq:variance_estimator} to calculate the estimate of the variance.

\section{Asymptotic Normality} \label{section:asumptoticnormality}

Now we prove the asymptotic normality of the estimator we described in the previous section. Like estimation section, we start with discussing the first stage of the estimation process in Section \ref{subsection:stage_1}

\subsection{Stage 1} \label{subsection:stage_1}

The normality of our estimator of $\beta_t$ depends on the estimation of the multiplicative inverse of the generalized propensity score (MIGPS) $\alpha(t,X)$ and the conditional expectation function $\gamma(t,X)$. We provide conditions for $\cL_2$ convergence rates for our estimators of the nuisance parameters.

 We adapt the assumptions from \cite{chernozhukov2016locally} to our setting. 

\begin{assumption} (Bounded dictionary) 
\label{assump:bounded_basis}
There exists a $C$ such that with probability one 
\begin{equation}
    \max_{1 \leq j \leq p} |b_j(X)| \leq C
\end{equation}
\end{assumption}

This assumption could also be weakened to allow for the bound on the basis functions to be an increasing function of the sample size $B_n$ rather than a constant $C$.

The next two assumptions - Assumptions \ref{assump:alpha_slow_rate} and \ref{assump:alpha_fast_rate} control the complexity of the true function $\alpha_0(t,X)$. Intuitively, the less complex $\alpha_0(t,X)$, the faster we can estimate it. Therefore the complexity of $\alpha_0(t,X)$ governs the convergence rate $\|\hat{\alpha}_0(t,X) - {\alpha}(t,X)\|$. When using the estimator we can assume either Assumption \ref{assump:alpha_slow_rate} or \ref{assump:alpha_fast_rate} - whichever we find plausible. If we use Assumption  \ref{assump:alpha_slow_rate} we will get a slower rate, and if we use Assumption \ref{assump:alpha_fast_rate} we will get a faster rate. We call the complexity in Assumption \ref{assump:alpha_slow_rate} the dense regime and in Assumption \ref{assump:alpha_fast_rate} the sparse regime.  

We will prove the convergence rates under the two different regimes in Lemmas \ref{lemma:alpha_convergence_slow} and \ref{lemma:alpha_convergence_fast} respectively.

\begin{assumption} 
\label{assump:alpha_slow_rate} (Dense regime)
For every $n$, there exists some $p \times 1$  sequence of coefficients $\rho_n$ and a positive constant $C < \infty$ such that $|\rho_n|_{1} \leq C$ and $\|\alpha_0 - b' \rho_n \|^2 = O(\sqrt{\frac{\ln(p)}{hn}} + h^2)$
\end{assumption}

Assumption \ref{assump:alpha_slow_rate} says that $\alpha_0$ can be approximated by our dictionary $b$. An example of when this assumption holds is when  $\alpha_0$ is a linear combination of the elements of the $b$ dictionary.  The specific $\epsilon_n$ rate assumed here follows from Lemma \ref{lemma:q_and_m_convergence}.

\begin{assumption} 
\label{assump:alpha_fast_rate} (Sparse regime)
Assume that the following hold.
\begin{enumerate}
    \item There exists $C, \xi >0$ such that for all $\bar{s}$ with  $\bar{s} \leq C(\sqrt{\frac{\ln(p)}{hn}} + h^2)^{\frac{-2}{(1 + 2 \xi)}}$  there is a $\bar{\rho} \in \R^p$ with $|\bar{\rho}|_1$ and $\bar{s}$ nonzero elements s.t. 
    
    \begin{equation}
        \| \alpha_0 - b' \bar{\rho} \| \leq C (\bar{s})^{- \xi}
    \end{equation}
    
    \item $Q = \E[f(t|X) b(X) b(X)']$ is nonsingular and has the largest eigenvalue uniformly bounded in $n$ 
    
    \item  for $\rho=\bar{\rho}$  and $\rho=\arg\min_{\rho}\{\left\Vert \alpha_{0}-b^{\prime}\bar{\rho}\right\Vert
^{2}+2r_{L}\sum_{j=1}^{p}\left\vert \rho_{j}\right\vert \}$\textit{
there is }$k>3$ such such that 

\begin{equation}
\label{eq:restricted_eigenvalue_condition}
    \inf_{\{\delta:\delta\neq0,\sum_{j\in\mathcal{J}_{\rho_{}}^{c}}|\delta
_{j}|\leq k\sum_{j\in\mathcal{J}_{\rho_{}}}|\delta_{j}|\}}\frac
{\delta^{\prime}Q\delta}{\sum_{j\in\mathcal{J}_{\rho_{}}}\delta_{j}^{2}}>0
\end{equation}

where $\mathcal{J}_{\rho_{}} = support(p)$

\end{enumerate}
\end{assumption}

Part 3 of assumption \ref{assump:alpha_fast_rate} is a population version of the restricted eigenvalue condition of \cite{bickel2009simultaneous} as adapted in \cite{chernozhukov2018automatic}. A clear introduction of the restricted eigenvalue condition is given in \cite{tibshirani2016closer}. 
 
To give some intuition about this part of assumption \ref{assump:alpha_fast_rate}, we relate it to the ``no perfect multicollinearity'' assumption of classic linear regression. When doing a linear regression of an outcome variable $Y$ on a covariate matrix $X$, we have to make sure that we don't have muticolinarity in order to ensure that matrix $\E[X^{T}X]$ is invertible. The matrix will only be invertible if none of the eigenvalues of the matrix are zero. We do not need such a strong condition in our case for matrix $Q$. Given that we are assuming a form of sparsity, we only require the invertability for the sub-matrices of $Q$ that we are considering.

\begin{assumption}(Regularization) \label{assump:regularization}
\begin{equation}
    r_{n} = a_n (\sqrt{\frac{\ln(p)}{hn}} + h^2) \text{ for some } a_n \rightarrow \infty
\end{equation}
\end{assumption}

To satisfy this assumption we set $a_n = \ln(\ln(n))$, following \cite{chatterjee2015prediction}. In practice, we pick a data-driven following \cite{chernozhukov2018automatic}.

\begin{lemma}
\label{lemma:q_and_m_convergence}

If Assumption \ref{assump:bounded_basis} holds for all $j$ then for $Q = \E[f(t|X) b(X) b(X)']$, 

\begin{equation}
    |\hat{Q}_{\ell} - Q |_{\infty} = O_p(\sqrt{\frac{\ln(p)}{hn}} + h^2) 
\end{equation}

\end{lemma}

The proof for Lemma \ref{lemma:q_and_m_convergence} is given in Appendix \ref{proof:q_and_m_convergence}. We denote the rate proven here as $\epsilon_n  = \sqrt{\frac{\ln(p)}{hn}} + h^2$ for the rest of the paper. This rate impacts the rates at which the nuisance parameters can estimated. In many examples in the Auto-DML literature - like the binary treatment case with IV - the rate $\epsilon = \sqrt{\frac{\ln(p)}{n}}$. In our case we have a slower rate with an additional $h^2$ because we are estimating at a fixed point $t$.

\begin{lemma} 
\label{lemma:alpha_convergence_slow} (Dense regime)
If assumption  and \ref{assump:alpha_slow_rate} holds then for any $r_L$ such that $\epsilon_n = o(r_L)$

    \label{lemma:alpha_convergence}
    \begin{equation}
        \label{eq:alpha_convergence_slow}
        \|\hat{\alpha}_{\ell}(t, X_i) - \alpha_0(t, X_i) \| = O_p(\sqrt{r_L})
    \end{equation}
\end{lemma}
    
The assumption that $\epsilon_n = o(r_L)$ means that our regularization parameter $r_L$ must go to zero slightly slower than $\epsilon_n$, which in our case $\epsilon_n  = O_p(\sqrt{\frac{\ln(p)}{hn}} + h^2)$.  One way to impose that $\epsilon_n = o(r_L)$ is to set $r_L$ proportional to $\epsilon_n\ln(\ln(n))$ in large samples. The $\ln(\ln(n))$ term does not impact the rates. 

Hence  

   \begin{equation}
        \|\hat{\alpha}_{\ell(t, X)} - \alpha_0(t, X) \| = O_p(\sqrt{(\sqrt{\frac{\ln(p)}{hn}} + h^2)\ln(\ln(n))})
    \end{equation}

The proof for Lemma \ref{lemma:alpha_convergence_slow} is given in Appendix \ref{proof:alpha_convergence_fast}

\begin{lemma}
\label{lemma:alpha_convergence_fast} (Sparse regime)
If assumption \ref{assump:alpha_fast_rate} holds, then we get a faster rate
        \begin{equation}
        \label{eq:alpha_convergence_fast}
        \|\hat{\alpha}_{\ell}(t, X) - \alpha_0(t, X) \| = O_p(\epsilon_{n}^\frac{-1 \xi}{(1 + 2 \xi)}r_L)
    \end{equation}
\end{lemma}

The proof for Lemma \ref{lemma:alpha_convergence_fast} is given in Appendix \ref{proof:alpha_convergence_fast}

\begin{assumption} 
\label{assump:regressionrate}(Regression rate)

For each  $\ell=1,...,L \text{ and for any } t \in \mathcal{T}$  

\begin{equation}
    \| \hat{\gamma}(t, X) - \gamma_0(t, X)\| = O_p(n^{-d_{\gamma}})
\end{equation}

\begin{enumerate}
    \item in the dense regime, $d_{\gamma} \in (\frac{1}{5}, \frac{1}{2})$ 
    \item in the sparse regime, $d_\gamma \in (\frac{4}{ 10} - \frac{4 \xi}{5 + 10 \xi}, \frac{1}{2})$
\end{enumerate}
\end{assumption}

These regime-specific bounds on $d_{\gamma}$ are sufficient conditions for the DML product condition given in Corollary \ref{corollary:rate_condition_interaction}. This controls the interaction remainder 

\begin{corollary} \label{corollary:rate_condition_interaction}
Under assumptions  \ref{assump:regularization}, \ref{assump:regressionrate} and - either assumption \ref{assump:alpha_slow_rate} OR \ref{assump:alpha_fast_rate}

\begin{equation}
    \|\hat{\alpha}(t, X) - \alpha_0(t, X)\| \|\hat{\gamma}(t, X) - \gamma_0(t, X)\| = o_p((hn)^{-\frac{1}{2}})
\end{equation}
\end{corollary}

 Corollary \ref{corollary:rate_condition_interaction} is proven in Appendix \ref{proof:rate_condition_interaction}. In this corollary we see the trade off in the error permitted in estimating our two nuisance parameters - the regression function $\gamma_0(t, X)$ and the MIGPS $\alpha_0(t,X)$.

\subsection{Stage 2} \label{subsection:stage_2}

To prove the normality of our estimator we follow the structure of Lemma 15 of \cite{chernozhukov2016locally}.

\begin{theorem}(asymptotic linearity) 
\label{theor:asymptotic_linearity}
Given Assumptions \ref{assump:identification}, \ref{assump:bounded_basis}, and either \ref{assump:alpha_slow_rate} or \ref{assump:alpha_fast_rate}, and  \ref{assump:regularization}, \ref{assump:regressionrate}
\begin{equation}
\sqrt{nh}(\hat{\psi}(\beta_0)) = \sqrt{\frac{h}{n}} \sum_{i = 1}^n {\psi}(W, \beta_0, \gamma_0, \alpha_0) + o_p(1)       
\end{equation}
\end{theorem}

Proofs following the structure of Lemma 15 from \cite{chernozhukov2016locally} require three different types of assumptions in order to prove asymptotic normality require three types of assumptions. First are the mild mean square consistency conditions  - which are satisfied given the rate conditions given in Lemmas \ref{lemma:alpha_convergence_slow} and  \ref{lemma:alpha_convergence_fast} Assumption \ref{assump:regressionrate}. Second is an assumption that the controls the interaction of the nuisance parameters, that is controlled by Corollary \ref{corollary:rate_condition_interaction}. Lastly, there is an assumption that controls that average of the double robustness term $\psi$ n our case as $h$ goes to zero this reminder goes to zero

\begin{theorem}(asymptotic normality) 
\label{theor:asymptotic_normality}
Let the same assumptions hold. Let $h \rightarrow 0$, $nh \rightarrow \infty$

\begin{equation}
\sqrt{nh}(\hat{\beta}_t - \beta_{0t} - h^2 B_t) \xrightarrow[]{d} N(0, V), \quad \hat{V}_t \xrightarrow[]{p} V_t 
\end{equation}

Where 

\begin{equation}
    V_t = \E[var[\gamma(t, X)]/f(t|X)]\int K(u)^2d_u
\end{equation}

and 
\begin{equation}
    B_t = \E[\frac{1}{2} \frac{\partial^2 }{\partial t^2} \E[Y| T = t, X] + \frac{\partial}{\partial t} \E[Y| T = t, X] \frac{\partial}{\partial t}f_{T|X}(t|X)/f_{T|X}(t|X) ] \int_{- \infty}^{\infty} u^2 k(u)d_u
\end{equation}

\end{theorem}

The $B_t$ is a bias term due to the fact the density estimators have with a non-zero asymptotic bias when one uses an optimal bandwidth $h$ - as explained in section 2.14 of \cite{hansen2009lecture}. We could make $B_t$ smaller by selecting a sub-optimal bandwidth, but then our estimator would have a slower convergence rate, so we will avoid this. We derive the asymptotic bias and variance terms in Lemmas \ref{lemma:asymptotic_bias} and \ref{lemma:asymptotic_var} in the appendix.

\section{Numerical examples} \label{section:numerical_examples}

\subsection{Simulation study}  \label{subsection:simulation study}

In this section we give simulation results for our estimator, and compare our estimator to the estimator proposed by \cite{colangelo2020double}. In their paper \cite{colangelo2020double} give simulation results to showcase the performance of their estimator under a specific DGP, and we use the same setup for our simulations. 

\subsubsection{DGP}

Consider two independent standard normal noise variables $\nu$ and $\epsilon$

\begin{equation}
    \nu \sim  N (0, 1), \quad \epsilon \sim N (0, 1) \quad
\end{equation}

Create 100 covariates $X$, and which are distributed normally also with mean zero and standard deviation 1, and are correlated with another according to the covariance matrix $\Sigma$. The $diag(\Sigma) = 1$, and $(i,j)$-entry $\Sigma_{ij} = 0.5$ for $|i-j| = 1$ and $\Sigma_{ij} = 0$ for $|i-j| > 1$ for $i,j = 1,\cdots,100$ 

\begin{equation}
    X=(X_1,\cdots,X_{100})'\sim N(0,\Sigma)
\end{equation}

Treatment variable $T$ is a function of noise $\nu$, the covariates $X$, and a vector of parameters $\theta$. The $j-th$ element of $\theta$ is $\theta_j = 1/j^2$. $\Phi$ denotes the CDF of $N(0,1)$. 
 
 \begin{equation}
     T=\Phi(3X'\theta)+0.75 \nu 
 \end{equation}

This means that our generalized propensity score is defined as
\begin{equation}
    f(t| X = x) = \frac{1}{.75 \sqrt{2\pi}}e^{-\frac{1}{2} (\frac{T - \Phi(3X'\theta)}{.75})^{2}}.
\end{equation}

Lastly our outcome variable,
\begin{equation}
    Y=1.2T+1.2X'\theta+T^2+TX_1 + \epsilon
\end{equation}

Thus the potential outcome $Y(t) = 1.2t+1.2X'\theta+ t^2 + tX_1 + \epsilon$. The parameter of interest in the simulations is the average dose response function at $t = 0$, i.e., $\beta_0 = E[Y(0)] = 0$.

We want to point out an important detail about the DGP. Because $\Phi$ is bounded between the values of 0 and 1, our $f(t| X = x)$ can only vary between 0.22 and 0.53. Therefore we expect estimators to perform well, since there are no small propensity scores that have to be estimated precisely in order to prevent large bias. 

The estimator from \cite{colangelo2020double} is 
\begin{equation}
    \hat{\beta}_t = \frac{1}{n} \sum_{\ell=1}^{L} \sum_{i \in \ell} \hat{\gamma}_{\ell}(t,X_i) +  K_h(T_i - t) {\frac{1}{\hat{f}_{T|X}(t|X_i)_{\ell}}}(Y_i - \hat{\gamma}_{\ell}(t, X_i))
\end{equation}

In comparison to ours 
\begin{equation}
       \hat{\beta}_t = \frac{1}{n} \sum_{\ell=1}^{L} \sum_{i \in \ell} \hat{\gamma}_{\ell}(t,X_i) +  K_h(T_i - t) \hat{\alpha_{\ell}}(t, X_i)(Y_i - \hat{\gamma}_{\ell}(t, X_i))
\end{equation}

A very key component to both of these estimators is the kernel $K$. We present results using both epanechnikov and gaussian kernels. We used the code from \cite{colangelo2020double} papers to produce Table 1. 

The first column gives the samples size of the simulation $n$. The second column gives the number of folds that were used in cross fitting $L$. Third column gives $c_h$ the size of the bin-width parameter used for kernel. The bin-width formula is given in equation \eqref{eq:bin-width}. In this equation $\sigma_T$ is the standard deviation of the treatment variable.   
\begin{equation}
\label{eq:bin-width}
h = c_h \sigma_T n^{-0.2}
\end{equation}

We can see that the table using the epanechnikov kernel in comparison to the gaussian kernel leads to lower bias and root mean square error (RMSE), and higher coverage. However, since the epanechnikov kernel is bounded, it is doing trimming implicitly, which as we discussed in the introduction is something we want to avoid as we could possibly lead to inference on a different population. The results of the \cite{colangelo2020double} estimator with a gaussian kernel are given in Table \ref{table:gaussian_sims_ying_ying}. We see that with high $c_h$ the estimator with the gaussian kernel does particularly badly - with the estimator losing proper coverage.

The simulation results for our estimator are given in Table \ref{table:gaussian_sims}. We use the gaussian kernel in order to avoid implicit trimming.  The $c_h$ is the kernel bin-width parameter like above. The $c_l$ is a parameter for the lasso regularization parameter $r_L$; the formula for $r_L$ is given by
 
\begin{equation}
\label{eq:regularization_parameter}
    r_L = c_l \Phi^{-1}(1 - .1/(2p))/ \sqrt{n}
\end{equation}

\begin{table}[h]
    \begin{subtable}[h]{0.45\textwidth}
    \label{table:epa_sims}
        \centering
\begin{tabular}{|c|c|c|l|l|l|}
\hline
\textbf{n} & \textbf{L} & \textbf{c}$_h$ & Bias & RMSE & Coverage \\ \hline
\multirow{5}{*}{\textbf{500}} & \multirow{5}{*}{\textbf{5}} & \textbf{0.5} & -0.111 & 2.737 & 0.959 \\ \cline{3-6} 
 &  & \textbf{1} & 0.029 & 0.135 & 0.943 \\ \cline{3-6} 
 &  & \textbf{1.5} & 0.053 & 0.125 & 0.938 \\ \hline
\multirow{5}{*}{\textbf{1000}} & \multirow{5}{*}{\textbf{5}} & \textbf{0.5} & 0.001 & 0.134 & 0.938 \\ \cline{3-6} 
 &  & \textbf{1} & 0.013 & 0.099 & 0.933 \\ \cline{3-6} 
 &  & \textbf{1.5} & 0.031 & 0.092 & 0.933 \\ \hline
\end{tabular}
       \caption{Epanechnikov kernel}
       \label{table:gaussian_sims_ying_ying}
    \end{subtable}
    \hfill
    \begin{subtable}[h]{0.45\textwidth}
        \centering
        \begin{tabular}{|c|c|c|l|l|l|}
\hline
\textbf{n} & \textbf{L} & \textbf{c}$_h$ & \multicolumn{1}{c|}{{Bias}} & \multicolumn{1}{c|}{{RMSE}} & \multicolumn{1}{c|}{{Coverage}} \\ \hline
\multirow{5}{*}{\textbf{500}} & \multirow{5}{*}{\textbf{5}} & \textbf{0.5} & 0.114 & 6.953 & 0.970 \\ \cline{3-6} 
 &  & \textbf{1} & 0.072 & 0.233 & 0.866 \\ \cline{3-6} 
 &  & \textbf{1.5} & 0.166 & 0.224 & 0.548 \\ \hline
\multirow{5}{*}{\textbf{1000}} & \multirow{5}{*}{\textbf{5}} & \textbf{0.5} & -0.008 & 0.148 & 0.945 \\ \cline{3-6} 
 &  & \textbf{1} & 0.040 & 0.102 & 0.885 \\ \cline{3-6} 
 &  & \textbf{1.5} & 0.114 & 0.158 & 0.553 \\ \hline
\end{tabular}
        \caption{Gaussian kernel}
        \label{tab:week2}
     \end{subtable}
    \caption{Colangelo and Lee Simulations}
    \label{table:colaangelo_lee_sims}
\end{table}

\begin{table}[ht]
\footnotesize
\centering
\begin{tabular}{rrrrrrrrr}
  \hline
N & L & $c_h$ & $c_l$ & bias & rmse & coverage \\ 
  \hline
500 &  5 & 0.5 & 0.50 & 0.0176 & 0.1158 & 0.9484 \\ 
  500 &  5 & 0.5 & 0.75 & 0.0151 & 0.1221 & 0.9495 \\ 
  500 &  5 & 0.5 & 1.00 & 0.0278 & 0.1171 & 0.9435 \\ 
  500 &  5 & 0.5 & 1.25 & 0.0195 & 0.1163 & 0.9478 \\ 
  500 &  5 & 1.0 & 0.25 & -1.1232 & 19.6210  & 0.9484 \\ 
  500 &  5 & 1.0 & 0.50 & 0.1126 & 0.1693 &  0.8531 \\ 
  500 &  5 & 1.0 & 0.75 & 0.0977 & 0.1464 &  0.8535 \\ 
  500 &  5 & 1.0 & 1.00 & 0.1061 & 0.1409  & 0.7896 \\ 
  500 &  5 & 1.0 & 1.25 & 0.1004 & 0.1464 &0.8477 \\ 
  500 &  5 & 1.5 & 0.25 & 0.8785 & 16.1961 &  0.9465 \\ 
  500 &  5 & 1.5 & 0.75 & 0.2141 & 0.2372 &  0.4483 \\ 
  500 &  5 & 1.5 & 1.00 & 0.2098 & 0.2312 &  0.4258 \\ 
  500 &  5 & 1.5 & 1.25 & 0.2026 & 0.2275 &  0.5014 \\ 
  1000 &  5 & 0.5 & 0.25 & 0.0069 & 0.1046 & 0.9478 \\ 
  1000 &  5 & 0.5 & 0.50 & 0.0146 & 0.0971 &  0.9472 \\ 
  1000 &  5 & 0.5 & 0.75 & 0.0125 & 0.0901 &  0.9479 \\ 
  1000 &  5 & 0.5 & 1.00 & 0.0133 & 0.0900 &  0.9486 \\ 
  1000 &  5 & 0.5 & 1.25 & 0.0089 & 0.0897 &  0.9486 \\ 
  1000 &  5 & 1.0 & 0.25 & 0.0966 & 0.1607 &  0.8867 \\ 
  1000 &  5 & 1.0 & 0.50 & 0.0745 & 0.1059 & 0.8324 \\ 
  1000 &  5 & 1.0 & 0.75 & 0.0821 & 0.1142 & 0.8211 \\ 
  1000 &  5 & 1.0 & 1.00 & 0.0759 & 0.1036 & 0.8162 \\ 
  1000 &  5 & 1.0 & 1.25 & 0.0782 & 0.1059 &  0.8114 \\ 
  1000 &  5 & 1.5 & 0.25 & 0.2348 & 0.2468 & 0.1289 \\ 
  1000 &  5 & 1.5 & 0.50 & 0.2049 & 0.2166 & 0.1690 \\ 
  1000 &  5 & 1.5 & 0.75 & 0.1802 & 0.1948 &  0.3192 \\ 
  1000 &  5 & 1.5 & 1.00 & 0.1723 & 0.1854 &  0.2852 \\ 
  1000 &  5 & 1.5 & 1.25 & 0.1650 & 0.1774 &  0.2847 \\ 
   \hline
\end{tabular}
\caption{simulation results gaussian} 
\label{table:gaussian_sims}
\end{table}

 We see that our results are pretty sensitive to the choice of $c_l$, and we shall work to reduce this sensitivity in future versions. For the reasonable $c_l$ values in most cases our estimator improves upon the Colangelo and Lee estimator with a gaussian kernel. This is especially true for the simulations with the smaller sample size of n = 500 --- there the RMSE is cut in half. With the larger sample size of n=1000 our estimators can decrease RMSE by about a third. 

It is important to note the both estimators perform  badly with large bin-widths --- at $c_h$ = 1.5 we see that our estimators no longer have valid coverage. Therefore the choice of the bin-width is also very important for the performance of estimators of this type.

\section{Conclusion} \label{section:conclusion}

In this paper, we presented a new estimator of continuous treatment effects and proved that it is asymptotically normal. Our estimator uses a new debiasing method that draws from both the DML and ADML literatures. We estimate the MIGPS – which is the critical component of the debiasing term – directly, as opposed to current methods that estimate the propensity of treatment and then invert it. Our direct estimation theoretically gives us improved numerical stability and nice automatic balancing properties. Empirically in Monte Carlo simulations, we find that our estimator decreases RMSE – up to 50\% in certain specifications – compared to current methods. In future versions of our paper we plan to make three key improvements. First, on the theoretical side, we plan to explain more intuition about our estimator's balancing properties. Second, on the empirical side, we want to improve the procedure that selects the optimal regularization parameter. Lastly, we also want to apply the estimator to a real world empirical example - specifically we plant to estimate the spatial gender wage gap as was studied in \cite{liu_geography_2020}.

\bibliographystyle{apalike}
\bibliography{main}

\appendix

\section{Stage 1}

\subsection{MIGPS}  \label{subsection:weightedols}

\paragraph{Proof of Lemma \ref{lemma:approximate_weighted_ols}}

\begin{proof}

 We want to  find $\hat{\rho}$ such that 
 \begin{equation}
     \hat{\alpha} = b(X)\hat{\rho} = \frac{1}{f_{T|X}(t|X)} 
 \end{equation}

Solve the weighted ols problem weight $w$

\begin{equation}
  \hat{\rho} = \argmin \E[w(\frac{1}{f(t|X)} - \rho'b(X))^2] 
 \end{equation}

 \begin{equation}
  \hat{\rho} = \argmin - 2\rho \E[w\frac{b(X)}{f(t|X)}] + \rho' \E [w b(X)b(X)'] \rho
\end{equation}

Set $ w = K_h(T - t)$

Then 

\begin{align}
    M  &= \E(\frac{K_h(t - t)b(X)}{f(t|X)}) = 
    \int_X \int_T \frac{K_h(T - t)b(X)}{f(t|X)} f(X,T) d_X d_T  \\
    &=  \int_X \int_T \frac{K_h(T - t)b(X)}{f(t|X)} f(T| X) f(X) d_X d_T \\
    &=  \int_X \bigg( \int_T K_h(T - t)f(T| X)  d_T \bigg) \frac{b(X)}{f(t|X)}  f(X) d_X 
\end{align}

We have $$\bigg( \int_T K_h(T - t)f(T| X)  d_T \bigg)  =\E_{T|X}[K_h(T - t)] = \E[K_h(T - t)|X] = f(t|X) + R(h)$$

If we can argue that $R(h)$ is small then 

\begin{align}
    M  &= \int_X f(t|X) \frac{b(X)}{f(t|X)}  f(X) d_X \\
    &= \int_X  b(X) f(X) d_X \\
    &= \E(  b(X)) \\
\end{align}

We make a similar argument for $Q$

Then we get 

$$\hat{M} = \frac{1}{n} \sum_{i = 1}^{n} b(X_i)$$

$$\hat{Q} = \frac{1}{n} \sum_{i = 1}^{n} K_h(T_i - t) b(X_i) b(X_i)'$$ 

Then 

$$\hat{\alpha}_{L} = b(X_i)\hat{\rho}_{L} $$
\end{proof}

\paragraph{Proof of Lemma \ref{lemma:q_and_m_convergence}} 

\begin{proof}
\label{proof:q_and_m_convergence}

We want to show $$|\hat{Q} - Q|_{\infty} = O_p(\sqrt{\frac{\ln(p)}{hn}} + h^2)$$

By adding and subtracting terms we have 
\begin{align}
\hat{Q} - Q = \hat{Q} &- \frac{1}{n} \sum b(X)b(X)'\E(K_h(T- t|X))\\ 
&+ \frac{1}{n} \sum b(X)b(X)'\E(K_h(T- t|X)) - \frac{1}{n} \sum b(X)b(X)'f(t|X)\\ &+ \frac{1}{n} \sum b(X)b(X)'f(t|X)  - Q  
\end{align}

Applying triangle inequality

\begin{align}
|\hat{Q} - Q|_{\infty} &\leq 
\underbrace{|\hat{Q} - \frac{1}{n} \sum b(X)b(X)'\E(K_h(T- t|X))|_{\infty}}_\textrm{term 1} \\ 
&+ \underbrace{|\frac{1}{n} \sum b(X)b(X)'\E(K_h(T- t|X)) - \frac{1}{n} \sum b(X)b(X)'f(t|X)|_{\infty}}_\textrm{term 2}  \\
&+ \underbrace{|\frac{1}{n} \sum b(X)b(X)'f(t|X)  - Q|_{\infty}}_\textrm{term 3}
\end{align}

We will bound the three terms on the right hand side

\begin{enumerate}
    \item

    The bound for the first term will follow from an argument from Van der varrt. We will show that  
$$|\hat{Q} - \frac{1}{n} \sum b(X)b(X)'\E(K_h(T- t|X))|_{\infty} = O_p(\sqrt{\frac{\ln(p)}{hn}}) $$

Let define our matrix $$A = \frac{1}{n}\sum_{i=1}^nb(X_i)b(X_i)'\bigg(K_h(T_i - t)  - \E[K_h(T_i - t)|X_i]\bigg)$$. 

Matrix $A$ has elements 

$$A_{j,k} = \frac{1}{n}\sum_{i=1}^n b_j(X_i)b_k(X_i)'\bigg(K_h(T_i - t)  - \E[K_h(T_i - t)|X_i]\bigg)$$

Each $A_{j,k}$ is an mean zero empirical process we will denote by $\mathbb{G}_n f_{jk}$ where

$$f_{jk} = b_{j}(X_i)b_{k}(X_i) \bigg(K_h(T_i - t)- \E[K_h(T_i - t)| X_i]\bigg)$$

For each $\mathbb{G}_n f_{jk}$ we will apply the following result from Van der Vaart 1998. 

For any bounded, measurable function $f$, then for every $x > 0 $

$$P_p(|\mathbb{G}_nf| > x) \leq 2 \text{exp} \bigg( -\frac{1}{4} \frac{x^2}{Pf^2 + x \|f\|_{\infty}/\sqrt{n}}\bigg)$$

 plugging in our $f_{jk}$

\begin{align}
  P_p(|\mathbb{G}_nf_{jk}| > x)  &= P_p(|\frac{1}{n}\sum_{i}b_j(X_i)b_k(X_i)'\bigg((K_h(T_i - t))  - \E[K_h(T_i - t|X_i)]\bigg)| > x) \\
  &\leq  2 \text{exp} \bigg( -\frac{1}{4} \frac{x^2}{Pf_{jk}^2 + x \|f_{jk}\|_{\infty}/\sqrt{n}}\bigg)
\end{align}

we have that $Pf_{jk}^2 < Ch^{-1}$ and $\|f_{jk}\|_{\infty} < Ch^{-1}$ $\forall f_{jk}$, so plugging that in

\begin{align}
  P_p(|\mathbb{G}_nf| > x)  &\leq  2 \text{exp} \bigg( -\frac{1}{4} \frac{x^2}{Ch^{-1} + x Ch^{-1}/\sqrt{n}}\bigg)
\end{align}

let $ x = \sqrt{n} t$

\begin{align}
  &\leq  2 \text{exp} \bigg( -\frac{1}{4} \frac{hnt^2}{C(1 + t)}\bigg)
\end{align}

Therefore 

\begin{align}
 P (|A|_{\infty} > t) &\leq  C p^2\text{exp} \bigg( - \frac{hnt^2}{(1 + t)}\bigg)
\end{align}

Let $t = \sqrt{\log(1/\epsilon)} \sqrt{\frac{\log(Cp^2)}{nh}}$

\begin{align}
 P (|A|_{\infty} > \sqrt{\log(1/\epsilon)} \sqrt{\frac{\log(C p^2)}{nh}}) &\leq  \epsilon \exp(-\frac{1}{1 + t}) \leq \epsilon \bigg)
\end{align}

Hence 

$$|A|_{\infty} = O_p(\sqrt{\frac{\ln(p)}{h^2 n}})$$

Which means that for every $\epsilon > 0$ there exists a $C_{\epsilon}$ such that 

$$\p(|A|_{\infty} \leq C_{\epsilon}\sqrt{\frac{\ln(p)}{h^2 n}} ) < \epsilon $$

    \item 
    
    The bound for the second term follows from a taylor expansion argument. We will show that  
    
    $$|\frac{1}{n} \sum b(X)b(X)'\E(K_h(T- t|X)) - \frac{1}{n} \sum b(X)b(X)'f(t|X)|_{\infty} = O_p(h^2)$$

    To show this second bound we will first prove that
    
    $$|E[K_h| X] - f(t|X)| \leq Ch^2$$

    Let $v = \frac{T - t}{h}$, therefore $ T = t + vh$

a mean value expansion of $f(t + vh|X)$ at $v = 0$ is 

$$f(t + vh|X) = f(t|X) + f'(t|X)vh + v^2h^2f''(t + \Tilde{v}h| X)$$

Where the last term is zero because we assume that $f''$ is bounded

\begin{align}
|\E[K_h(T - t))  - f(t|X)| X]| &= |\int \frac{1}{h}K\bigg(\frac{T - t}{h}\bigg) f(T|X) - f_0(t|X) d_T| \\
&= |\int K\big(v\big) f(t + vh|X) - f_0(t|X) d_v|
\end{align}

now we plug in our mean value expansion

\begin{align}
&= |\int K\big(v\big) \bigg( f(t|X) + f'(t|X)vh + v^2h^2f''(T + \Tilde{v}h| X) \bigg) - f_0(t|X) d_v| \\
&= \\
&\leq Ch^2
\end{align}

    \item 
    
    the third term follows from an application of Hoeffdings inequality

    $$|\frac{1}{n} \sum b(X)b(X)'f(t|X)  - Q|_{\infty} = O_p(\sqrt{\frac{\ln(p)}{hn}})$$
    
\end{enumerate}

Therefore combining these three terms we get  

$$|\hat{Q} - Q|_{\infty} \leq O_p(\sqrt{\frac{\ln(p)}{hn}}) + O_p(h^2) + O_p(\sqrt{\frac{\ln(p)}{hn}})$$

$$= O_p(\sqrt{\frac{\ln(p)}{hn}} + h^2)$$

\end{proof}

\paragraph{Proof of Lemma \ref{lemma:alpha_convergence_slow}}  
\begin{proof}
\label{proof:alpha_convergence_slow}
Using Lemma $\ref{lemma:q_and_m_convergence}$ and Assumptions \ref{assump:bounded_basis} and \ref{assump:alpha_slow_rate} we can apply the results of Theorem 1 \cite{chernozhukov2018automatic}. For their theorem three assumptions must be satisfied 

\begin{enumerate}
    \item basis functions bounded, which holds from Assumption \ref{assump:bounded_basis}
    \item rates $\epsilon^M_n$ and $\epsilon^G_n$, which follow from Lemma $\ref{lemma:q_and_m_convergence}$
    \item and the sparse approximation rate in Assumption \ref{assump:alpha_slow_rate} 
\end{enumerate}

Hence for any $r_L$ such that $o(r_L) = \sqrt{\frac{\ln(p)}{hn}} + h^2 $

       $$ \|\hat{\alpha}_{\ell} - \alpha_0 \|^2 = O_p(\sqrt{\frac{\ln(p)}{h n}} + h^2)\ln(\ln(n)))$$

\end{proof}

\paragraph{Proof of Lemma \ref{lemma:alpha_convergence_fast}} 
\begin{proof}
\label{proof:alpha_convergence_fast}
Using Lemma $\ref{lemma:q_and_m_convergence}$ and Assumptions \ref{assump:bounded_basis} and \ref{assump:alpha_fast_rate} we can apply the results of Theorem 3 \cite{chernozhukov2018automatic}. For $\epsilon = o(r_L)$

       $$ \|\hat{\alpha}_{\ell} - \alpha_0 \|^2 = O_p\bigg(\big(\sqrt{\frac{\ln(p)}{h n}} + h^2)\ln(\ln(n))\big) \bigg)$$

\end{proof}

\section{Justification of rates} \label{appendix:Justification of rates}

\paragraph{Proof of Corollary \ref{corollary:rate_condition_interaction}}
\label{proof:rate_condition_interaction}
This section outlines the condition will be sufficient for the interaction term in the remainder decomposition to go to zero \eqref{eq:rate_condition_interaction}. 

\begin{equation}
\label{eq:rate_condition_interaction}
     \|\hat{\alpha}_L(t,X) - {\alpha}_0(t, X)  \| \|\hat{\gamma}_L(t,X) - {\gamma}_0(t, X)  \| = o_p((nh)^{-\frac{1}{2}})
\end{equation}

\begin{enumerate}
    \item in the dense regime

Recall for the dense regime we have shown in Lemma \ref{lemma:alpha_convergence} that $\|\hat{\alpha}(t, X_i) - \alpha_0(t, X_i)\| = O_p(\sqrt{(\sqrt{\frac{\ln(p)}{h n}} + h^2)\ln(\ln(n))})$

We require $\hat{\gamma}(t, X)$ to be estimated at some uniform mean square rate $-d_{\gamma}$ such that. 

\begin{equation}
 \begin{aligned}
         &\sqrt{nh} \|\hat{\alpha}(t, X_i) - \alpha_0(t, X_i)\| \|\hat{\gamma}(t, X_i) - \gamma_0(t, X_i)\| \\
         &= O_p\bigg(n^{1/2} h^{1/2}\sqrt{(\sqrt{\frac{\ln(p)}{h n}} + h^2)\ln(\ln(n))} n^{-d_{\gamma}}\bigg) \overset{p}{\longrightarrow}0 
 \end{aligned}
\end{equation}

We need a $d_{\gamma}$ value so that the exponent on the $n$ term is less than zero. We solve for our required $d_{\gamma}$

\begin{equation}
\begin{aligned}
 \bigg(n^{1/2} h^{1/2}\sqrt{(\sqrt{\frac{\ln(p)}{h n}} + h^2)\ln(\ln(n))} n^{-d_{\gamma}}\bigg)
\end{aligned}
\end{equation}

For algebraic simplicity square the entire expression. We also ignore the $\ln(p)$ and $\ln(\ln(n))$ terms. In order to to balance bias and variance set $h = n^{-1/5}$ \cite{wasserman2006all}. 

\begin{equation}
\begin{aligned}
& n h (\sqrt{\frac{\ln(p)}{h n}} + h^2) n^{-2d_{\gamma}} \\
 &=n h h^{-1/2} n^{-1/2} n^{-2d_{\gamma}}  + n h h^2 n^{-2d_{\gamma}} \\
 &=n^{1/2-2d_{\gamma} } h^{1/2}  + n^{1 -2d_{\gamma}}h^{3} \\
  &=n^{1/2-2d_{\gamma} } n^{-1/10}  + n^{1 -2d_{\gamma}}h^{-6/10} \\
  &=n^{4/10 -2d_{\gamma}}  + n^{4/10 -2d_{\gamma}} \\
&= 2n^{4/10 -2d_{\gamma}} 
\end{aligned}
\end{equation}

So we need 

\begin{equation}
\begin{aligned}
     4/10 -2d_{\gamma} &< 0 \\
     2/10 &< d_{\gamma}  \\
     1/5 &< d_{\gamma}  \\
\end{aligned}
\end{equation}

    \item in the sparse regime
    
    From Lemma \ref{lemma:alpha_convergence_fast} we have the following rate for $\alpha$: $\|\hat{\alpha}_{\ell}(t, X_i) - \alpha_0(t, X_i) \| = O_p(\epsilon_{n}^\frac{-1 \xi}{(1 + 2 \xi)}r_L)$. We use the same procedure as done in the dense regime above. 

We need a $d_{\gamma}$ value so that the exponent on the $n$ term is less than zero. We solve for our required $d_{\gamma}$

\begin{equation}
\begin{aligned}
 \bigg(n^{1/2} h^{1/2}\bigg(\sqrt{\frac{\ln(p)}{h n}}+ h^2 \bigg)^{\frac{-1}{1 + 2 \xi}} (\sqrt{\frac{\ln(p)}{h n}}+ h^2)\ln(\ln(n)) n^{-d_{\gamma}}\bigg)
\end{aligned}
\end{equation}

Again we ignore the $\ln(p)$ and $\ln(\ln(n))$ terms.

\begin{equation}
\begin{aligned}
 \bigg(n^{1/2} h^{1/2}\bigg( h^{-1/2} n^{-1/2}+ h^2 \bigg)^{\frac{2 \xi}{1 + 2 \xi}}  n^{-d_{\gamma}}\bigg)
\end{aligned}
\end{equation}

Again set $h = n^{-1/5}$. 

\begin{equation}
\begin{aligned}
& n^{1/2} h^{1/2}\bigg( h^{-1/2} n^{-1/2}+ h^2 \bigg)^{\frac{2 \xi}{1 + 2 \xi}}  n^{-d_{\gamma}  } \\
&=n^{1/2} n^{-1/10}\bigg( 2n^{-2/5}  \bigg)^{\frac{2 \xi}{1 + 2 \xi}}  n^{-d_{\gamma}  }  \\
 &= n^{4/10}\bigg( 2n^{-2/5}  \bigg)^{\frac{2 \xi}{1 + 2 \xi}}  n^{-d_{\gamma}  } 
\end{aligned}
\end{equation}

So we need 
\begin{equation}
\begin{aligned}
     \frac{4}{ 10} - \frac{2}{5}\times\frac{2 \xi}{1 + 2 \xi}  -d_{\gamma}   &< 0 \\
 \frac{4}{ 10} - \frac{4 \xi}{5 + 10 \xi} &< d_{\gamma} \\
\end{aligned}
\end{equation}

\end{enumerate}

\section{Normality Proof} \label{section:normalityproof}

\paragraph{Proof of Theorem \ref{theor:asymptotic_linearity} (asymptotic linearity)} 

We want to show 

\begin{equation}
\label{eq:asymptotic_linearity}
\sqrt{nh}(\hat{\psi}(\beta_0)) = \sqrt{\frac{h}{n}} \sum_{i = 1}^n {\psi}(W_i, \beta_0, \gamma_0, \alpha_0) + o_p(1)    
\end{equation}

\begin{proof}
\label{proof:asymptotic_linearity}

To prove \eqref{eq:asymptotic_linearity} we will show 

\begin{equation}
   \frac{\sqrt{h}}{\sqrt{n}} \sum_{\ell = 1}^L  \sum_{i \in I_{\ell}} \psi(W_i, \tilde{\beta}_{\ell}, \hat{\gamma}_{\ell}, \hat{\alpha}_{\ell}) - {\psi}(W_i, \beta_0, \gamma_0, \alpha_0) \xrightarrow{p} 0 
\end{equation}

Consider one fold of the data $\ell$

\begin{equation}
\label{eq:one_fold}
{\psi}(W, \tilde{\beta}_{\ell}, \hat{\gamma}_{\ell}, \hat{\alpha}_{\ell}) - {\psi}(W, \beta_0, \gamma_0, \alpha_0) = g(W, \beta_0, \hat{\gamma}_{\ell}) - g(W, \beta_0, \gamma_0) + \phi(W, \tilde{\beta}_{\ell}, \hat{\gamma}_{\ell}, \hat{\alpha}_{\ell}) - \phi(W, \beta_0, \gamma_0, \alpha_0)   
\end{equation}

We go through the remainder decomposition that used for our result. 

Add and subtract to the right hand side of \eqref{eq:one_fold} $\pm \textcolor{blue}{\phi(W, \beta_0, \hat{\gamma}_{\ell}, \alpha_0)}$, $\pm \textcolor{WildStrawberry}{\phi(W, \tilde{\beta}_{\ell}, \gamma_0, \hat{\alpha}_{\ell})}$, and $\pm \textcolor{Emerald}{\phi(W, \beta_0, \gamma_0, \alpha_0)}$ and rearrange the terms.

\begin{align}
&= g(W, \beta_0, \hat{\gamma}_{\ell}) - g(W, \beta_0, \gamma_0)  \tag{1} \\
&+ \textcolor{blue}{\phi(W, \beta_0, \hat{\gamma}_{\ell}, \alpha_0)} - \textcolor{Emerald}{\phi(W, \beta_0, \gamma_0, \alpha_0)} \tag{2}\\  
&+ \textcolor{WildStrawberry}{\phi(W, \tilde{\beta}_{\ell}, \gamma_0, \hat{\alpha}_{\ell})} - \phi(W, \beta_0, {\gamma}_0,{\alpha}_0) \tag{3} \\
&+ \phi(W, \tilde{\beta}_{\ell}, \hat{\gamma}_{\ell}, \hat{\alpha}_{\ell}) - \textcolor{WildStrawberry}{\phi(W, \tilde{\beta}_{\ell}, \gamma_0, \hat{\alpha}_{\ell})} - \textcolor{blue}{\phi(W, \beta_0, \hat{\gamma}_{\ell}, \alpha_0)} +  \textcolor{Emerald}{\phi(W, \beta_0, \gamma_0, \alpha_0)} \tag{$\Delta$}
\end{align}

Let $W_{\ell}^c$ denote the observations not in fold $I_{\ell}$. Next subtract out and add back the means (conditional on $W_{\ell}^c$) of the  of the first three terms. After this, there are six core remainder terms left. 

\begin{align}
&= g(W, \beta_0, \hat{\gamma}_{\ell}) - g(W, \beta_0, \gamma_0)  - \textcolor{Maroon}{\E[ g(W, \beta_0, \hat{\gamma}_{\ell}) - g(W, \beta_0, \gamma_0) |W_{\ell^c}]} \tag{$R_1$} \\
&+ {\phi(W, \beta_0, \hat{\gamma}_{\ell}, \alpha_0)} - {\phi(W, \beta_0, \gamma_0, \alpha_0)} - \textcolor{Maroon}{\E[{\phi(W, \beta_0, \hat{\gamma}_{\ell}, \alpha_0)} - {\phi(W, \beta_0, \gamma_0, \alpha_0)} | W_{\ell}^c]} \tag{$R_2$}\\  
&+ {\phi(W, \tilde{\beta}_{\ell}, \gamma_0, \hat{\alpha}_{\ell})} - \phi(W, \beta_0, {\gamma}_0,{\alpha}_0) - \textcolor{MidnightBlue}{\E[{\phi(W, \tilde{\beta}_{\ell}, \gamma_0, \hat{\alpha}_{\ell})} - \phi(W, \beta_0, {\gamma}_0,{\alpha}_0) | W_{\ell}^c]} \tag{$R_3$} \\
&+ \phi(W, \tilde{\beta}_{\ell}, \hat{\gamma}_{\ell}, \hat{\alpha}_{\ell}) - {\phi(W, \tilde{\beta}_{\ell}, \gamma_0, \hat{\alpha}_{\ell})} - {\phi(W, \beta_0, \hat{\gamma}_{\ell}, \alpha_0)} +  {\phi(W, \beta_0, \gamma_0, \alpha_0) } \tag{$\Delta$} \\
&+ \textcolor{Maroon}{\E[g(W, \beta_0, \hat{\gamma}_{\ell}) - g(W, \beta_0, \gamma_0) |W_{\ell}^c] + \E[\phi(W, \beta_0, \hat{\gamma}_{\ell}, \alpha_0)] - \phi(W, \beta_0, \gamma_0, \alpha_0) | W_{\ell}^c]} \tag{$E_{1+2}$} \\
&+ \textcolor{MidnightBlue}{\E[{\phi(W, \tilde{\beta}_{\ell}, \gamma_0, \hat{\alpha}_{\ell})} - \phi(W, \beta_0, {\gamma}_0,{\alpha}_0) | W_{\ell}^c]} \tag{$E_{3}$}
\end{align}

This expression above is the full remainder expansion. We will show that 

$$\sum_{\ell = 1}^L \frac{\sqrt{h}}{\sqrt{n}} \sum_{i \in I_{\ell}} R_{1 \ell i} + R_{2 \ell i} + R_{3 \ell i} + \Delta_{ \ell i} + E_{1+2, \ell i} +E_{3, \ell i} \xrightarrow{p} 0$$

This will be proved by showing each of the six remainder terms goes to zero in probability

\begin{enumerate}
    \item Reminder $R_1$
    
    Though we assume rate conditions for our nuisance parameters in this paper, reminder $R_1$ goes to zero given weaker mild mean square convergence rate assumptions, and so we prove our result with this weaker condition.

   \begin{align}
       R_{1 \ell i} =& g(W_i, \beta_0, \hat{\gamma}_{\ell}) - g(W_i, \beta_0, \gamma_0)  - {\E[ g(W_i, \beta_0, \hat{\gamma}_{\ell}) - g(W, \beta_0, \gamma_0) |W_{\ell}^c]} \\
       &= \hat{\gamma}_{\ell}(t, X_i) - {\gamma}_{0}(t, X_i) - \E[\hat{\gamma}_{\ell}(t, X_i) - {\gamma}_{0}(t, X_i) | W_{\ell}^C ]
   \end{align}
   
First note that for a fold $\ell$ of our data we have 
   
 \begin{equation}
 \label{eq:r_1_conditional_to_zero}
     \begin{aligned}
             & \E_{W_{\ell}}[\Bigg( \sqrt{\frac{h}{n}} \sum_{i \in I_{\ell}} \bigg(\hat{\gamma}_{\ell}(t, X) - {\gamma}_{0}(t, X) - \E[\hat{\gamma}_{\ell}(t, X) - {\gamma}_{0}(t, X) | W_{\ell}^C ]\bigg) \Bigg)^2 | W_{\ell}^c] \\
       &\leq \frac{h n_{\ell}}{n} Var \bigg(\hat{\gamma}_{\ell}(t, X) - {\gamma}_{0}(t, X) | W_{\ell}^c\bigg) \\
       &\leq h \E_{W_{\ell}}\bigg( (\hat{\gamma}_{\ell}(t, X) - {\gamma}_{0}(t, X))^2 | W_{\ell}^c]\bigg)  \xrightarrow{p} 0 \text{  by Corollary \ref{corollary:rate_condition_interaction}}
     \end{aligned}
 \end{equation}

Apply the conditional markov inequality 

\begin{equation}
  P( \sqrt{\frac{h}{n}} \sum_{i \in I_{\ell}} R_{1 \ell i} > a |W_{\ell}^C) \leq P( |\sqrt{\frac{h}{n}} \sum_{i \in I_{\ell}} R_{1 \ell i}| > a|W_{\ell}^C) \leq \frac{\E[|\sqrt{\frac{h}{n}} \sum_{i \in I_{\ell}} \Delta_{ \ell i}|W_{\ell}^C]|}{a} \xrightarrow{p} 0  
\end{equation}

Now we show that we can move from a conditional probability statement $\bigg(P( \sqrt{\frac{h}{n}} \sum_{i \in I_{\ell}} R_{1 \ell i} > a |W_{\ell}^C) \xrightarrow{p} 0  \bigg)$  to an unconditional probability statement $\bigg(P( \sqrt{\frac{h}{n}} \sum_{i \in I_{\ell}} R_{1 \ell i} > a) \xrightarrow{p} 0  \bigg)$ 
 
By \eqref{eq:r_1_conditional_to_zero} $P( \sqrt{\frac{h}{n}} \sum_{i \in I_{\ell}} R_{1 \ell i} > a |W_{\ell}^C) \xrightarrow{p} 0$. Apply the DCT, since $P( \sqrt{\frac{h}{n}} \sum_{i \in I_{\ell}} R_{1 \ell i} > a |W_{\ell}^C)$ is bounded by 1, and conclude that $\E[P( \sqrt{\frac{h}{n}} \sum_{i \in I_{\ell}} R_{1 \ell i} > a |W_{\ell}^C)] \xrightarrow{} 0 $. Therefore since 
 \begin{equation}
  \label{eq:conditional_to_unconditional}
   \E[P( \sqrt{\frac{h}{n}} \sum_{i \in I_{\ell}} R_{1 \ell i} > a |W_{\ell}^C)] = P( \sqrt{\frac{h}{n}} \sum_{i \in I_{\ell}} R_{1 \ell i} > a)
 \end{equation}
 \begin{equation}
     \E[P( \sqrt{\frac{h}{n}} \sum_{i \in I_{\ell}} R_{1 \ell i} > a |W_{\ell}^C)] \rightarrow 0 \implies \sqrt{\frac{h}{n}} \sum_{i \in I_{\ell}} R_{1 \ell i} \xrightarrow{p} 0  
 \end{equation}

 Hence $\sqrt{\frac{h}{n}} \sum_{i \in I_{\ell}} R_{1 \ell i} \xrightarrow{p} 0 $ for each fold of the data $I_{\ell}$ 
 
 Summing across all folds we can conclude $\sum_{\ell = 1}^L \frac{\sqrt{h}}{\sqrt{n}} \sum_{i \in I_{\ell}} R_{1 \ell i } \xrightarrow{p} 0$

    \item Reminder $R_2$
    
    $\sum_{\ell = 1}^L \frac{\sqrt{h}}{\sqrt{n}} \sum_{i \in I_{\ell}} R_{2 \ell i } \xrightarrow{p} 0$  follows by same argument as $R_1$
    
    \item remainder $R_3$
    
    $\sum_{\ell = 1}^L \frac{\sqrt{h}}{\sqrt{n}} \sum_{i \in I_{\ell}} R_{3 \ell i}  \xrightarrow{p} 0$ follows by the same argument as $R_1 $
    
    \item remainder $\Delta$
    
    \begin{align}
        \Delta_{ \ell i} &= \phi(W_i, \tilde{\beta}_{\ell}, \hat{\gamma}_{\ell}, \hat{\alpha}_{\ell}) - {\phi(W_i, \tilde{\beta}_{\ell}, \gamma_0, \hat{\alpha}_{\ell})} - {\phi(W_i, \beta_0, \hat{\gamma}_{\ell}, \alpha_0)} +  {\phi(W_i, \beta_0, \gamma_0, \alpha_0)} \\
        &= K_h(T_i - t)\bigg[\hat{\alpha}(t, X_i)(Y_i - \hat{\gamma}(t, X)) -
       {\alpha}_0(t, X_i)(Y_i - \hat{\gamma}(t, X)) \\
       &\quad \quad \quad \quad \quad - \hat{\alpha}(t, X)(Y - {\gamma}_0(t, X_i)) + 
        \alpha_0(t, X_i)(Y_i - \gamma_0(t, X_i)) \bigg] \\
        &=  K_h(T_i - t)(\hat{\alpha}_{\ell}(t, X_i) - \alpha_0(t, X))(\hat{\gamma}_{\ell}(t, X_i) - \gamma_0(t, X_i)) 
    \end{align}
    
    Therefore 
    
\begin{equation}
\label{{eq:delta_remainder}}
        \sqrt{\frac{h}{n}} \sum_{i \in I_{\ell}} \Delta_{ \ell i} = \sqrt{\frac{h}{n}} \sum_{i \in I_{\ell}} K_h(T_i - t)(\hat{\alpha}_{\ell}(t, X_i) - \alpha_0(t, X_i))(\hat{\gamma}_{\ell}(t, X_i) - \gamma_0(t, X_i)) 
\end{equation}
    
    By the condition markov inequality we have that 
    
    $$P( \sqrt{\frac{h}{n}} \sum_{i \in I_{\ell}} \Delta_{ \ell i} > a |W_{\ell}^C) \leq P( |\sqrt{\frac{h}{n}} \sum_{i \in I_{\ell}} \Delta_{ \ell i}| > a|W_{\ell}^C) \leq \frac{\E[|\sqrt{\frac{h}{n}} \sum_{i \in I_{\ell}} \Delta_{ \ell i}||W_{\ell}^C]|}{a} $$
    
    So to bound the $\Delta$ reminder I will first show that $\E[|\sqrt{\frac{h}{n}} \sum_{i \in I_{\ell}} \Delta_{ \ell i}|W_{\ell}^C] \xrightarrow{p} 0$
    
\begin{equation}
\begin{aligned}
& \E[|\sqrt{\frac{h}{n}} \sum_{i \in I_{\ell}} \Delta_{ \ell i}| W_{\ell}^C] \\
&= \E[|\sqrt{\frac{h}{n}} \sum_{i \in I_{\ell}} K_h(T_i - t)(\hat{\alpha}_{\ell}(t, X_i) - \alpha_0(t, X_i))(\hat{\gamma}_{\ell}(t, X_i) - \gamma_0(t, X_i))| | W_{\ell}^C] \\
&= \sqrt{hn} \int_{\cX} \int_{\cT}| K_h(T_i - t)(\hat{\alpha}_{\ell}(t, X_i) - \alpha_0(t, X_i))(\hat{\gamma}_{\ell}(t, X_i) - \gamma_0(t, X_i))| f(X_i, T_i) d_{T_i} d_{X_i}\\
&= \sqrt{hn} \int_{\cX} |(\hat{\alpha}_{\ell}(t, X_i) - \alpha_0(t, X_i))(\hat{\gamma}_{\ell}(t, X_i) - \gamma_0(t, X_i))|\int_{\cT} |K_h(T_i - t) f(T_i|X_i)|   f(X_i) d_{T_i}  d_{X_i} \\
&= \sqrt{hn} \int_{\cX} |(\hat{\alpha}_{\ell}(t, X_i) - \alpha_0(t, X_i))(\hat{\gamma}_{\ell}(t, X_i) - \gamma_0(t, X_i))| \int_{\cT} |\frac{1}{h}K(\frac{T_i - t}{h}) f(T_i|X_i)|  f(X_i) d_{T_i}  d_{X_i} \\
\end{aligned}
\end{equation}

Where in the first inequality we substituted the definition of the delta remainder, and in the next inequality we switch the expectation with the sum, and used the fact that the observations are iid.  Next, substitute $f(T_i|X_i)$ by its second-order Taylor expansion around point $t$. Start by substituting $u_i = \frac{T_i - t}{h}$. 

\begin{equation}
\begin{aligned}
&= \sqrt{hn} \int_{\cX} |(\hat{\alpha}_{\ell}(t, X_i) - \alpha_0(t, X_i))(\hat{\gamma}_{\ell}(t, X_i) - \gamma_0(t, X_i))| \int_{\cU} |K(u_i) f(t - h u_i|X_i)  |  f(X_i) d_{u_i}  d_{X_i} \\
&= \sqrt{hn} \int_{\cX} |(\hat{\alpha}_{\ell}(t, X_i) - \alpha_0(t, X_i))(\hat{\gamma}_{\ell}(t, X_i) - \gamma_0(t, X_i))| \times \\
& \quad \quad \quad \quad \int_{\cU} |K(u_i) \big[f(t|X_i) - h u f'(t|X_i) + \frac{h^2 u_i^2}{2}f''(t|X_i) + \cdots  \big]  |  f(X_i) d_{u_i}  d_{X_i}
\end{aligned}
\end{equation}

Evaluate the inner integral with respect to $u_i$, recalling that $\int_{u_i} K(u_i) d_{u_i} = 1$.

\begin{equation}
\begin{aligned}
 &\leq \sqrt{hn} \int_{\cX} |(\hat{\alpha}_{\ell}(t, X_i) - \alpha_0(t, X_i))(\hat{\gamma}_{\ell}(t, X_i) - \gamma_0(t, X_i))| \times \\
& \quad \quad \quad \quad |f(t|X_i) + \frac{1}{2}h^2f''(t|X_i)\int_{u_i}u_i^2K(u_i) d_{u_i} + \cdots | f(X_i) d_{u_i}  d_{X_i} \\
    \end{aligned}
\end{equation}
 Recall $\int_{u_i} u_i^2 K(u_i) d_{u_i} \leq C$  and that $f''(t|X) \leq C$ 

\begin{equation}
\begin{aligned}
 &\leq \sqrt{hn} \int_{\cX} |(\hat{\alpha}_{\ell}(t, X_i) - \alpha_0(t, X_i))(\hat{\gamma}_{\ell}(t, X_i) - \gamma_0(t, X_i))| \times  |f(t|X_i) + o(h^2) | f(X_i)  d_{X_i} \\
 & \text{why does the o term become op and how am I able to take it outside} \\
&= \sqrt{hn} \int_{\cX} f(t|X_i)  |(\hat{\alpha}_{\ell}(t, X_i) - \alpha_0(t, X_i))(\hat{\gamma}_{\ell}(t, X_i)- \gamma_0(t, X_i))|  f(X_i) d_{X_i} + o_p(\sqrt{nh} h^2)
    \end{aligned}
\end{equation}
By the Cauchy-Swartz inequality and the fact that $nh^4 \rightarrow C$

\begin{equation}
\small
    \begin{aligned}
    &\leq \sqrt{hn}  \bigg(\int_{\cX} f(t|X)(\hat{\gamma}_{\ell}(t, X_i)- \gamma_0(t, X_i))^2  f(X_i) d_{X_i}\bigg)^{.5}  \bigg(\int_{\cX} f(t|X)(\hat{\alpha}_{\ell}(t, X_i)- \alpha_0(t, X_i))^2  f(X_i) d_{X_i}\bigg)^{.5} \\
    &\quad \quad \quad \quad + o_p(1)\\
     &= \sqrt{hn}  \|(\hat{\gamma}_{\ell}(t, X_i)- \gamma_0(t, X_i))\| \| \hat{\alpha}_{\ell}(t, X_i)- \alpha_0(t, X_i))\|  + o_p(1)       
    \end{aligned}
\end{equation}
 
   In order to control the $\|(\hat{\gamma}_{\ell}(t, X_i)- \gamma_0(t, X_i))\|$ and $\| \hat{\alpha}_{\ell}(t, X_i)- \alpha_0(t, X_i))\|$ we use the mean square convergence rates from Assumption \ref{assump:regressionrate}, and the rates assumed for $\hat{\alpha}$

\begin{equation}
    \begin{aligned}
     & \sqrt{hn}  \|(\hat{\gamma}_{\ell}(t, X_i)- \gamma_0(t, X_i))\| \| \hat{\alpha}_{\ell}(t, X_i)- \alpha_0(t, X_i))\|  + o_p(1) \xrightarrow{p}  0   
    \end{aligned}
\end{equation}

Hence 

$$\E[|\sqrt{\frac{h}{n}} \sum_{i \in I_{\ell}} \Delta_{ \ell i}|| W_{\ell}^C] \xrightarrow{p} 0$$

Apply the same condition to unconditional argument used in equation \eqref{eq:conditional_to_unconditional} 

\begin{equation}
   \sum_{\ell = 1}^L \frac{\sqrt{h}}{\sqrt{n}} \sum_{i \in I_{\ell}} \Delta_{\ell i} \xrightarrow{p}  0
\end{equation}

    \item remainder $E_{1+ 2}$
    
    \begin{align}
        E_{1+ 2} &= \E[g(W, \beta_0, \hat{\gamma}_{\ell}) - g(W, \beta_0, \gamma_0) |W_{\ell}^c] + \E[\phi(W, \beta_0, \hat{\gamma}_{\ell}, \alpha_0)] - \phi(W, \beta_0, \gamma_0, \alpha_0) | W_{\ell}^c] \\
        &= \E[\psi(W, \beta_0, \hat{\gamma}_{\ell}) - \psi(W, \beta_0, \gamma_0) |W_{\ell}^c]
    \end{align}
    
    follows by the same argument as the $\Delta$ term

    \item remainder $E_{3}$
    
    This is the ``double robustness'' remainder. When bias corrected moments are double robust, which happens if and only if our moment function is an affine transformation of the first stage estimator, the ``double robustness'' remainder is zero. This is what happens in the binary treatment case. This double robustness holds in our case as $h \rightarrow 0 $  

    We can use the same argument as \cite{colangelo2020double} and conclude 
    
    $E_{3} = O_p\bigg((\| \hat{\gamma} - \gamma_0\| + \| \hat{\alpha} - \alpha_0\|) \sqrt{nh}h^2 \bigg) = o_p(1)$

\end{enumerate}

All six of the remainder terms go to zero in probability. Add the six remainder terms together and apply the triangle inequality 
\begin{equation}
    \sum_{\ell = 1}^L \frac{\sqrt{h}}{\sqrt{n}} \sum_{i \in I_{\ell}} R_{1 \ell i} + R_{2 \ell i} + R_{3 \ell i} + \Delta_{ \ell i} + E_{1+2, \ell i} +E_{3, \ell i} \xrightarrow{p} 0
\end{equation}

Hence we can conclude 
\begin{equation}
    \sqrt{nh}(\hat{\psi}(\beta_0)) = \sqrt{\frac{h}{n}} \sum_{i = 1}^n {\psi}(W, \beta_0, \gamma_0, \alpha_0) + o_p(1) 
\end{equation}

\begin{equation}
\label{eq:op_1_psi}
    \sqrt{nh}(\hat{\psi}(\beta_0) - {\psi}(W, \beta_0, \gamma_0, \alpha_0)) =   o_p(1) 
\end{equation}

\end{proof}

Now some helpful algebra 
\begin{corollary}

How do we move from the equation above 
\begin{equation}
    \sqrt{nh}(\hat{\psi}(\beta_0) - {\psi}(W, \beta_0, \gamma_0, \alpha_0)) =   o_p(1) 
\end{equation}

to the equation of interest

\begin{equation}
    \begin{aligned}
         \sqrt{nh}(\hat{\beta}_t  -  \beta_0)  = \sqrt{\frac{h}{n}} \sum_{\ell = 1}^L \sum_{i \in I_{\ell}}[{\gamma}_0(t, X_i) - \beta_0 + \frac{K_h(T_i - t)}{f(t|X_i)}(Y_i - {\gamma}_0(t, X_i))]) +  o_p(1)
    \end{aligned}
\end{equation}

Recall 

\begin{equation}
\begin{aligned}
\frac{1}{n} \sum_{\ell = 1}^L \sum_{i \in I_{\ell}} {\psi}_{i}(\beta_0) &= \frac{1}{n} \sum_{\ell = 1}^L \sum_{i \in I_{\ell}}g(W_i, \beta_0, {\gamma_0}) + \phi(W_i, \beta_0, \gamma_0, \alpha_0) \\
     &=  \frac{1}{n} \sum_{\ell = 1}^L \sum_{i \in I_{\ell}} {\gamma}_0(t, X_i) - \beta_0 + \frac{K_h(T_i - t)}{f(t|X_i)}(Y_i - {\gamma}_0(t, X_i)) \\
\end{aligned}
\end{equation}

and

\begin{equation}
\begin{aligned}
     \frac{1}{n} \sum_{\ell = 1}^L \sum_{i \in I_{\ell}} \hat{\psi}_{i \ell}(\beta_0) &=  \frac{1}{n} \sum_{\ell = 1}^L \sum_{i \in I_{\ell}} g(W_i, \beta_0, \hat{\gamma}) + \phi(W_i, \beta_0, \hat{\gamma}, \hat{\alpha}) \\
     &= \frac{1}{n} \sum_{\ell = 1}^L \sum_{i \in I_{\ell}} \hat{\gamma}(t, X_i) - \beta_0 + \frac{K_h(T - t)}{\hat{f}(t|X_i)}(Y_i - \hat{\gamma}(t, X_i)) \\
     &= \frac{1}{n} \sum_{\ell = 1}^L \sum_{i \in I_{\ell}} \hat{\gamma}(t, X_i)  + \frac{K_h(T_i - t)}{\hat{f}(t|X_i)}(Y_i - \hat{\gamma}(t, X_i)) - \beta_0 \\
    &= \hat{\beta}_t  - \beta_0 \\
\end{aligned}
\end{equation}

Where the last equality followed from the definition of the estimator \eqref{eq:estimator}

Hence 

\begin{equation}
    \begin{aligned}
       \frac{1}{n} \sum_{\ell = 1}^L \sum_{i \in I_{\ell}}  \hat{\psi}(\beta_0) - {\psi}(\beta_0) = \hat{\beta}_t  -  \beta_0 - \frac{1}{n} \sum_{\ell = 1}^L \sum_{i \in I_{\ell}} [{\gamma}_0(t, X_i) - \beta_0 + \frac{K_h(T_i - t)}{f(t|X_i)}(Y_i - {\gamma}_0(t, X_i))]
    \end{aligned}
\end{equation}

By \eqref{eq:op_1_psi}

\begin{equation}
    \begin{aligned}
         \sqrt{nh}(\hat{\beta}_t  -  \beta_0 - \frac{1}{n} \sum_{\ell = 1}^L \sum_{i \in I_{\ell}}[{\gamma}_0(t, X_i) - \beta_0 + \frac{K_h(T_i - t)}{f(t|X_i)}(Y - {\gamma}_0(t, X_i))]) = o_p(1)
    \end{aligned}
\end{equation}

Therefore 
\begin{equation}
\label{eq:op_1_beta}
    \begin{aligned}
         \sqrt{nh}(\hat{\beta}_t  -  \beta_0)  = \sqrt{\frac{h}{n}} \sum_{\ell = 1}^L \sum_{i \in I_{\ell}}[{\gamma}_0(t, X_i) - \beta_0 + \frac{K_h(T_i - t)}{f(t|X_i)}(Y_i - {\gamma}_0(t, X_i))]) +  o_p(1)
    \end{aligned}
\end{equation}

\end{corollary}

Now we are ready to study the bias 
\begin{lemma}
\label{lemma:asymptotic_bias}
$\hat{B_t}$
\end{lemma}

\begin{proof}

take expectations on both sides of \eqref{eq:op_1_beta}

\begin{equation}
    \begin{aligned}
         \E[(\hat{\beta}_t  -  \beta_0)]  &= \E[\frac{1}{n} \sum_{\ell = 1}^L \sum_{i \in I_{\ell}}[{\gamma}_0(t, X_i) - \beta_0 + \frac{K_h(T_i - t)}{f(t|X_i)}(Y_i - {\gamma}_0(t, X_i))]) +  o_p(1)]\\
         &=  \E[{\gamma}_0(t, X_i) - \beta_0 + \frac{K_h(T_i - t)}{f(t|X_i)}(Y_i - {\gamma}_0(t, X_i))]) +  o_p(1)]
    \end{aligned}
\end{equation}

Since $\E[{\gamma}_0(t, X_i)] = \beta_0 $

\begin{equation}
    \begin{aligned}
         \E[(\hat{\beta}_t  -  \beta_0)] 
         &=  \E[\frac{K_h(T_i - t)}{f(t|X_i)}(Y_i - {\gamma}_0(t, X_i))]) +  o_p(1)]
    \end{aligned}
\end{equation}

ignoring the $o_p(1)$ term for now, apply the law of iterated expectations. We now add subscripts on the expectation notation to keep straight what variables we are integrating over.  

\begin{equation}
    \begin{aligned}
         \E_{XT}[(\hat{\beta}_t  -  \beta_0)] 
         &=  \E_X[\E_T[\frac{K_h(T_i - t)}{f(t|X_i)}(Y_i - {\gamma}_0(t, X_i))]] \\
                  &=  \E_X[\frac{1}{f(t|X_i)}  \E_T[K_h(T_i - t)(Y_i - {\gamma}_0(t, X_i))| X_i]] \\
                    &=  \E_X[\frac{1}{f(t|X_i)}  \E_T[K_h(T_i - t)({\gamma}_0(T_i, X_i) - {\gamma}_0(t, X_i))| X_i]] \\
    \end{aligned}
\end{equation}

Now let us focus on the inner expectation from the line above 

\begin{equation}
    \begin{aligned}
         \E_T[K_h(T_i - t)(Y_i - {\gamma}_0(t, X_i))| X_i]] &= \int_{\cT} K_h(T_i - t)({\gamma}_0(T_i, X_i) - {\gamma}_0(t, X_i)) f(T_i| X_i) d_{T_i}  \\
         &= h^2 R(K) [\partial_t \gamma_0(t, X)  {\partial_t f(t|X)} + \frac{1}{2}\partial^2_t \gamma(t,X) f(t|X) ] + O(h^3)
    \end{aligned}
\end{equation}

Hence

\begin{equation}
    \E_{XT}[(\hat{\beta}_t  -  \beta_0)] = \E_{X}[\frac{1}{f(t|X)} h^2 R(K) [\partial_t \gamma_0(t, X)  {\partial_t f(t|X)} + \frac{1}{2}\partial^2_t \gamma(t,X) f(t|X) ] + O(h^3)]
\end{equation}

\end{proof}

Now we move onto the variance

\begin{lemma}
\label{lemma:asymptotic_var}
$\hat{V_t}$ 
\end{lemma}

\begin{proof}

Following from equation \eqref{eq:op_1_beta}, applying the CLT

\begin{equation}
    Var({\gamma}_0(t, X_i) - \beta_0 + \frac{K_h(T_i - t)}{f(t|X_i)}(Y_i - {\gamma}_0(t, X_i))]) +  o_p(1))
\end{equation}

\begin{equation}
\begin{aligned}
     \hat{V_t} &= \frac{1}{n} \sum_{\ell = 1}^L \sum_{i \in I_{\ell}} \bigg( g(W_i, \hat{\beta}, \hat{\gamma}_{\ell}) + \phi(W_i, \hat{\beta},  \hat{\gamma}_{\ell}, \hat{\alpha}_{\ell} ) \bigg)^2 \\
     &= \frac{1}{n} \sum_{\ell = 1}^L \sum_{i \in I_{\ell}} \bigg( \hat{\gamma}(t, X_i)  - \hat{\beta} + \frac{K_h(T_i - t)}{\hat{f}(t|X_i)}(Y_i - \hat{\gamma}(t, X_i))  \bigg)^2
\end{aligned}
\end{equation}

we use the same argument as $\hat{B_t}$ above to conclude 

\end{proof}

\paragraph{Proof for Theorem \ref{theor:asymptotic_normality}}

\begin{equation}
\sqrt{n}(\hat{\theta} - \theta_0) \xrightarrow[]{d} N(0, V), \quad \hat{V} \xrightarrow[]{p} V 
\end{equation}

Given the asymptotically linear result above, and our arguments for $\hat{B_t}$  and $\hat{V_t}$,  normality follows from the central limit theorem

\section{Empirical Examples } \label{section:more_emperical_examples}

Continuous treatment effects are estimated in many different sub-areas of economics. A few examples are included below \footnote{if the reader has suggestions of other examples of continuous treatment effects, please send to klosins@mit.edu - would be much appreciated}.

\begin{itemize}
    \item In the political economy literature \cite{cantoni2020precinct} study voting costs specifically using distance to polling location as the treatment of interest. 
    \item  \cite{deshpande2019screened} study how the increased travel times due to closings of Social Security Administration field offices impact disability insurance applications.  
     \item \cite{liu_geography_2020} study the spatial gender wage gap. They run a linear model to predict wages as a function covariates including gender as well as there variable of interest, commute time - a continuous variable that is the treatment of interest 
     \item In the immigration literature, often the ``percentage of imgrants'' in geographic area is a continuous variable of interest e.g.  \cite{borjas2013analytics}
     \item In the trade literature, distance is a common friction that is studied: For example \cite{brei2018distance} look at the effect of two distance measures  ``measuring the effect of cross-border distance relative to that of domestic distance'' as an obstacle to trade.

     \item \cite{diamond2019wants} treatment of interest is distance from low income housing. They look at ``spillovers of properties financed by the Low Income Housing Tax Credit (LIHTC) onto neighborhood residents''

     \item \cite{yagan2019employment} using longitudinal linked employer-employee data to isolate causal effects of Great Recession local shocks on employment. The parameter of interest is the causal effect on one's 2015 outcomes of living in 2007 in a local area that experienced a one-unit larger Great Recession shock.  
     
\end{itemize}

\end{document}